\documentclass[letterpaper, 10 pt, conference]{ieeeconf}
\IEEEoverridecommandlockouts
\usepackage{amsmath,amssymb}
\usepackage{mathtools}
\usepackage{bm}

\usepackage{amsthm} 
\usepackage{subfig}
\usepackage{graphicx}
\let\labelindent\relax
\usepackage{enumitem} 

\usepackage{xcolor}

\newtheorem{lemma}{Lemma}
\newtheorem{proposition}{Proposition}
\newtheorem{theorem}{Theorem}
\newtheorem{corollary}{Corollary}
\newtheorem{remark}{Remark}
\newtheorem{assumption}{Assumption}

\usepackage{color,soul}

\title{Nonlinear Density-Driven Optimal Control (D$^2$OC) for Multi-Agent Spatial Coverage via Sequential Convex Programming}
\author{Julian Martinez$^{1}$ and Kooktae Lee$^{1}$
	\thanks{*This work was supported by NSF CAREER Grant CMMI-DCSD-2638508.}
	\thanks{$^{1}$Julian Martinez and Kooktae Lee are with the Department of Mechanical and Aerospace Engineering, Texas Tech University, Lubbock, TX 79409, USA, email: mar85781@ttu.edu, kooktae.lee@ttu.edu.}  
}

\begin{document}

\maketitle

\begin{abstract}
This paper presents a nonlinear extension of Density-Driven Optimal Control (D$^2$OC) for multi-agent spatial coverage with prescribed density distributions. Rather than assigning individual target locations, D$^2$OC drives the collective spatial distribution of agents toward a desired density through a Wasserstein-based objective. We extend this framework to multi-step finite-horizon control for discrete-time control-affine nonlinear systems using sequential convex programming. At each control update, the nonlinear dynamics are locally linearized over the prediction horizon, yielding a strictly convex quadratic program that preserves the Wasserstein barycentric structure while directly incorporating input constraints. We further characterize the effect of constrained control deviations and nonlinear Taylor remainders on the accuracy of the local linear prediction, establishing an explicit finite-horizon error bound and a two-step specialization for receding-horizon implementation. The resulting method retains the decentralized, distribution-driven nature of D$^2$OC while providing a computationally efficient optimization procedure for nonlinear multi-agent systems. Simulations with unicycle and quadrotor teams show coverage performance comparable to nonlinear model predictive control, while substantially reducing computation time. These results demonstrate a tractable and theoretically characterized framework for density-driven spatial coverage under nonlinear dynamics.
\end{abstract}

\section{Introduction}
\label{sec:introduction}

Multi-agent spatial coverage aims to distribute autonomous agents over a spatial domain according to sensing, monitoring, or mission-dependent requirements. Classical coverage control methods, such as the Voronoi-based framework, partition the workspace into agent-specific Voronoi cells and steer agents toward their cell centroids \cite{cortes2004coverage}. While effective for geometric coverage, such methods do not directly address non-uniform agent distributions. This motivates formulating non-uniform coverage as a spatial distribution-matching problem.

Optimal transport (OT) provides a principled framework for comparing and matching spatial distributions through Wasserstein distances \cite{villani2009optimal,chen2021optimaltransport}. This perspective has also been incorporated into multi-agent coverage control, where the objective is formulated directly in terms of the discrepancy between the empirical distribution of agents and a prescribed spatial density, rather than through individually assigned target locations \cite{inoue2021otcc}. Building on this distribution-based viewpoint, Density-Driven Optimal Control (D$^2$OC) formulates non-uniform multi-agent coverage as a decentralized optimal-control problem that drives the empirical agent distribution toward a reference density through local sample selection, weight updates, and neighbor information exchange \cite{seo2025smcs,lee2025connectivity,seo2025tcst,lee2026density}.

The D$^2$OC framework was initially developed for linear time-invariant (LTI) dynamics and later extended to linear time-varying (LTV) dynamics \cite{lee2025connectivity,seo2025tcst}. Its extension to nonlinear systems, however, has so far been limited to a one-step prediction horizon. A multi-step nonlinear formulation must account for nonlinear agent dynamics while preserving the Wasserstein-based coverage objective, decentralized structure, and input constraints. 
NMPC provides a framework for nonlinear dynamics and constraints through repeated finite-horizon optimization \cite{mayne2000cmpc}, while distributed receding-horizon methods address computational and communication requirements in multi-agent systems \cite{dunbar2006distributed}.
Sequential convex programming (SCP), in contrast, locally approximates nonlinear dynamics and solves a sequence of convex subproblems, and has been applied to real-time nonlinear optimal control and successive convexification \cite{trandin2011rtscp,mao2016successive,mao2018tutorial}. This makes SCP a natural candidate for extending D$^2$OC to multi-step nonlinear coverage although the resulting local approximation introduces a discrepancy between predicted and actual trajectories that must be characterized.

In this paper, we extend D$^2$OC to discrete-time control-affine
nonlinear systems using sequentially linearized local models. Each
optimization step yields a strictly convex local Quadratic Program (QP) with input
constraints while retaining the Wasserstein-based coverage objective.
We further characterize constraint-induced control deviations and derive a finite-horizon bound on the resulting
linearization error.

Thus, the main contributions of this paper are summarized as follows:
1) We extend the Wasserstein-based D$^2$OC formulation to discrete-time control-affine nonlinear systems while preserving its decentralized coverage mechanism and distribution-matching objective;
2) We formulate each nonlinear D$^2$OC update as a strictly convex local QP via sequential convex programming, retaining the Wasserstein barycenter structure and actuator constraints while enabling efficient computation;
3) We quantify input-saturation effects through the deviation between constrained and unconstrained QP solutions and incorporate this deviation into the linearization-error analysis;
4) We establish a finite-horizon error bound accounting for control deviations and second-order Taylor remainders, with a receding-horizon specialization for first-input application.

The resulting framework combines the distribution-level objective of D$^2$OC with efficient local optimization and explicit prediction-error characterization for nonlinear multi-agent dynamics.
\section{Problem Setup}
\label{sec:problem_setup}
\noindent \textbf{Notation:} The sets of real numbers are denoted by \(\mathbb{R}\). The sets \(\mathbb{Z}_{>0}\) and \(\mathbb{Z}_{\geq 0} := \mathbb{Z}_{>0} \cup \{0\}\) denote positive and non-negative integers.  
The space \(\mathbb{R}^n\) represents \(n\)-dimensional column vectors. 
The Euclidean norm of a vector \(\mathbf{x}\) is written as \(\|\mathbf{x}\|\), and the transpose of a matrix \(A\) as \(A^{\top}\). The zero matrix \(\mathbf{0}_{m\times n} \in \mathbb{R}^{m\times n}\), the identity matrix \(\mathbf{I}_n \in \mathbb{R}^{n\times n}\), and the all-ones vector \(\mathbf{1}_n \in \mathbb{R}^n\) are denoted with their sizes as subscripts. 
The operator \(\mathrm{diag}([\cdot])\) constructs a block diagonal matrix from its arguments, while \(\operatorname{col}(\cdot)\) vertically concatenates its vector or matrix arguments. The Kronecker and Hadamard (elementwise) products are denoted by \(\otimes\) and \(\odot\), respectively.

\medskip

We consider a network of $n_a$ agents, where each agent is modeled by
a discrete-time control-affine nonlinear system
\begin{equation}
\begin{aligned}
x_{i,k+1}
&= f_i(x_{i,k}) + G_i(x_{i,k})u_{i,k},\\
y_{i,k}
&= C_i x_{i,k},
\end{aligned}
\label{eq:dyn}
\end{equation}
where $x_{i,k}\in\mathbb{R}^n$, $u_{i,k}\in\mathbb{R}^m$, and
$y_{i,k}\in\mathbb{R}^d$ denote the state, control input, and output,
respectively. The symbols
$f_i:\mathbb{R}^n\rightarrow\mathbb{R}^n$ and
$G_i:\mathbb{R}^n\rightarrow\mathbb{R}^{n\times m}$ describe the
nonlinear drift and input dynamics, respectively, and
$C_i\in\mathbb{R}^{d\times n}$ maps the state to the coverage output.
We assume that $f_i$ and $G_i$ are sufficiently smooth over the
region of interest.

The objective is to control the agents so that their empirical spatial
distribution approaches a prescribed reference distribution while
respecting input and communication constraints. We adopt the
Density-Driven Optimal Control (D$^2$OC) framework \cite{seo2025smcs,seo2025tcst,lee2025connectivity,lee2026density} for this purpose.
The original D$^2$OC formulation is based on a Wasserstein-distance
objective and admits a QP formulation for linear
time-invariant dynamics. In this work, this structure is used as the
starting point for extending D$^2$OC to control-affine nonlinear
systems. In particular, we retain the Wasserstein-based coverage
objective while replacing the globally linear prediction model with
a sequentially linearized local model.

For a discrete probability measure, let $\delta_y$ denote the Dirac
measure at $y$. We characterize the desired spatial distribution $\nu$ by a
finite set of sample points
$\{q_j\}_{j=1}^N\subset\mathbb{R}^d$ with associated weights
$\beta_j$, and the agent distribution by the empirical measure $\rho_k$ at time $k$ as follows:
\begin{equation}
\rho_k
:=
\frac{1}{k+1}
\sum_{t=0}^{k}
\left(
\frac{1}{n_a}
\sum_{i=1}^{n_a}\delta_{y_{i,t}}
\right),
\quad
\nu
:=
\sum_{j=1}^{N}\beta_j\delta_{q_j}.
\label{eq:distributions}
\end{equation}
The coverage objective is to reduce the discrepancy between $\rho_k$
and $\nu$, measured by the quadratic Wasserstein distance
\begin{equation}
\mathcal{W}_2^2(\mu,\nu)
=
\min_{\pi_{\ell j}}
\sum_{\ell=1}^{M}\sum_{j=1}^{N}
\pi_{\ell j}\|y_\ell-q_j\|^2,
\label{eq:W-dist}
\end{equation}
subject to the standard transport constraints
\[
\pi_{\ell j}\geq 0,\qquad
\sum_j\pi_{\ell j}=\alpha_\ell,\qquad
\sum_\ell\pi_{\ell j}=\beta_j.
\]

Rather than solving the global Wasserstein problem centrally, the
Density-Driven Optimal Control (D$^2$OC) framework operates through the
following decentralized cycle \cite{seo2025smcs,seo2025tcst,lee2026density}:

\begin{enumerate}[leftmargin=15pt]
    \item \textbf{Local sample selection and control:}
    each agent selects nearby sample points with nonzero remaining
    weights and computes a feasible control input to reduce its local
    Wasserstein cost.

    \item \textbf{Weight update:}
    after executing the control action, the agent decreases the
    remaining weights of the sample points it covers, recording
    its coverage progress.

    \item \textbf{Neighbor information exchange:}
    when agents communicate, they synchronize their sample weights using
    the information received from neighboring agents, avoiding redundant
    coverage.
\end{enumerate}

Thus, the first two stages are decentralized, while communication is
used only to maintain consistency of the coverage information. In this
work, we focus on the optimal-control step in Stage 1 and its extension
to nonlinear agent dynamics.

\section{Formulation of the D$^2$OC Cost Function via Quadratic Programming}
\label{sec:d2oc_qp}

We first recall how the Wasserstein objective is formulated in D$^2$OC.
Unlike the standard optimal transport problem, where the transport
plan $\pi_{\ell j}$ is the optimization variable, D$^2$OC uses the current agent
position to analytically determine the transport plan and 
identify where the agent should move to match the prescribed density.
Over a prediction horizon, these weights evolve with the predicted agent
position.

For agent $i$, let $\mathcal{S}_{i,k+h}$ denote the selected reference
samples at prediction step $k+h$, with corresponding transport weights
$\pi_{j,k+h}$. The local D$^2$OC cost is
\begin{equation}
\sum_{h=1}^{H}\mathcal{W}_i^2(k+h)
:=
\sum_{h=1}^{H}
\sum_{j\in\mathcal{S}{i,k+h}}
\pi_{j,k+h}
||y_{i,k+h}-q_j||^2.
\label{eq:local_wasserstein_cost}
\end{equation}

Defining the weighted barycenter at each prediction step as
\begin{equation}
\bar q_{i,k+h}
:=
\frac{
\sum_{j\in\mathcal{S}_{i,k+h}}
\pi_{j,k+h}q_j
}{
\sum_{j\in\mathcal{S}_{i,k+h}}
\pi_{j,k+h}
},
\label{eq:barycenter}
\end{equation}
the Wasserstein cost becomes a weighted least-squares tracking problem
with respect to the time-varying barycenter. Stacking the predicted
outputs and barycenters gives
\begin{equation}
\begin{aligned}
Y_i^{k|H}
&:=
\begin{bmatrix}
y_{i,k+1}^\top & \cdots &
y_{i,k+H}^\top
\end{bmatrix}^\top,\\
\Xi_i^{k|H}
&:=
\begin{bmatrix}
\bar q_{i,k+1}^\top & \cdots &
\bar q_{i,k+H}^\top
\end{bmatrix}^\top.
\end{aligned}
\label{eq:YQ}
\end{equation}

Using
$
\small
\boldsymbol{\Omega}_i^{k|H}
:=
\operatorname{diag}
\Bigg(
\sqrt{\sum_{j\in\mathcal{S}_{i,k+h}}
\pi_{j,k+h}},\mathbf I_d
\Bigg)_{h=1}^{H},
\label{eq:Omega}
$
the local Wasserstein objective can then be written as a weighted
least-squares tracking cost as follows.

\begin{proposition}[Wasserstein-to-barycenter reduction]
\label{prop:wasserstein_barycenter} (Proposition 1 in \cite{lee2025connectivity})
Let the transport weights $\pi_{j,k+h}$ be fixed over the prediction
horizon. Then,
\begin{equation}
\sum_{h=1}^{H}\mathcal{W}_i^2(k+h)
=
\left\|
\boldsymbol{\Omega}_i^{k|H}
\left(
Y_i^{k|H}-\Xi_i^{k|H}
\right)
\right\|^2
+\mathrm{const.}
\label{eq:wasserstein_reduction}
\end{equation}
\end{proposition}

For the LTI dynamics, the stacked output prediction is affine in the
stacked control input
\begin{equation}
Y_i^{k|H}
=
\Theta_iU_i^{k|H}+\Phi_i x_{i,k},
\label{eq:stacked_prediction}
\end{equation}
where
\begin{equation}
U_i^{k|H}
=
\begin{bmatrix}
u_{i,k}^\top & \cdots & u_{i,k+H-1}^\top
\end{bmatrix}^\top,
\end{equation}
\begin{equation}
\small
\Theta_i=
\begin{bmatrix}
C_iB_i & 0 & \cdots & 0\\
C_iA_iB_i & C_iB_i & \cdots & 0\\
\vdots & \vdots & \ddots & \vdots\\
C_iA_i^{H-1}B_i &
C_iA_i^{H-2}B_i & \cdots &
C_iB_i
\end{bmatrix},
\label{eq:Theta}
\end{equation}
\begin{equation}
\small
\Phi_i
=
\begin{bmatrix}
(C_iA_i)^\top &
(C_iA_i^2)^\top &
\cdots &
(C_iA_i^{H})^\top
\end{bmatrix}^{\top}.
\label{eq:Phi}
\end{equation}

Consequently, Proposition~\ref{prop:wasserstein_barycenter} establishes that the complex distribution matching problem can be effectively recast into a standard weighted barycentric tracking task together with the weight update stage at each time step introduced in Section II.

\section{Nonlinear Optimality Conditions and SCP-Based Local QP Formulation}
\label{sec:opt_nonlin}

Consider the discrete-time control-affine nonlinear dynamics in \eqref{eq:dyn}. For clarity, we omit the agent index $i$ throughout this section. Let $k$ denote the current decision time and $t\in{k,\ldots,k+H-1}$ the prediction-time index. The dynamics at each prediction step are given by
\[
    x_{t+1}=f(x_t)+G(x_t)u_t,\;\; y_t=Cx_t,
    \quad t=k,\ldots,k+H-1.
\]

Let the stage cost at time $t$ be
\begin{equation}
    \ell_t(x_t,u_t) := (y_t-\bar q_t)^\top Q_t (y_t-\bar q_t) + u_t^\top R\, u_t,
    \label{eq:stage_cost}
\end{equation}
where $Q_t\in\mathbb{R}^{d\times d}$ is symmetric and positive semidefinite and
$R\succ 0$.
Over a prediction horizon of length $H$, let the terminal cost be
\[
    \ell_{k+H}(x_{k+H})
    =(y_{k+H}-\bar q_{k+H})^\top Q_{k+H}(y_{k+H}-\bar q_{k+H}),
\]
if needed.

\subsection{PMP Characterization and NMPC Baseline}
\label{subsec:pmp_baseline}

The finite-horizon nonlinear optimal control problem can be characterized by discrete-time Pontryagin's Minimum Principle (PMP). For the control-affine dynamics
\eqref{eq:dyn}, any locally optimal solution satisfies the first-order necessary conditions:
\begin{subequations}\label{eq:PMP}
\begin{align}
x_{t+1}
&= f(x_t)+G(x_t)u_t,
\label{eq:PMP_dyn}\\
\lambda_t
&= 2C^\top Q_t(Cx_t-\bar q_t)
   +\frac{\partial \Psi}{\partial x}(x_t,u_t)^\top\lambda_{t+1},
\label{eq:PMP_costate}\\
\lambda_{k+H}
&= 2C^\top Q_{k+H}(Cx_{k+H}-\bar q_{k+H}),
\label{eq:PMP_terminal}
\end{align}
\end{subequations}
where 
\begin{equation}
\Psi(x,u):=f(x)+G(x)u.
\label{eq:Psi}
\end{equation}

Under box control constraints
$u_{\mathrm{min}} \le u_t \le u_{\mathrm{max}}$, the control satisfies the
pointwise Hamiltonian minimization condition
\begin{equation}
u_t^\star
\in
\arg\min_{u_{\mathrm{min}}\le u\le u_{\mathrm{max}}}
\left\{
u^\top R u
+\lambda_{t+1}^\top G(x_t)u
\right\}.
\label{eq:PMP_control_min}
\end{equation}
Equivalently, this condition follows from the KKT conditions of the
box-constrained quadratic problem. When $R$ is diagonal, its solution is
available in closed form as
\begin{equation}
u_t^\star
=
\operatorname{sat}_{[u_{\mathrm{min}}, u_{\mathrm{max}}]}
\left(
-\frac{1}{2}R^{-1}G(x_t)^\top\lambda_{t+1}
\right).
\label{eq:PMP_control}
\end{equation}
Thus, the saturation in \eqref{eq:PMP_control} is the exact pointwise
minimizer rather than a heuristic clipping operation.

These conditions formulate a two-point boundary-value problem (TPBVP) with
$x_k$ given and $\lambda_{k+H}$ specified at the terminal horizon. A standard
shooting approach treats the initial costate $\lambda_k = s$ as a decision
variable, propagating the state-costate dynamics to find a root of the terminal
residual
\begin{equation}
F(s) := \lambda_{k+H}(s)
- 2C^\top Q_{k+H}\big(Cx_{k+H}(s)-\bar q_{k+H}\big) = 0.
\label{eq:shooting_residual}
\end{equation}
While this PMP-based shooting method serves as our baseline NMPC reference,
the saturation operator in \eqref{eq:PMP_control} introduces active-set
switching and hence makes $F(s)$ generally only piecewise smooth, which can
complicate Newton-type root-finding and Jacobian evaluation.

To overcome these limitations, our primary computational strategy relies on
Sequential Convex Programming (SCP). Rather than solving the resulting
piecewise-smooth shooting problem, we locally linearize the dynamics around a
nominal trajectory and solve a sequence of convex QPs, with the input bounds
imposed directly as QP constraints. This provides a convenient numerical
formulation for handling input constraints while preserving the Wasserstein
tracking objective. The resulting local QP formulation is presented next.

\subsection{SCP-Based Local QP Reduction}
\label{subsec:qp_reduction}

The preceding PMP formulation provides the necessary optimality conditions
for the original nonlinear problem and motivates a shooting-based NMPC
procedure as a reference method. The proposed approach instead avoids
solving the nonlinear TPBVP directly. We employ a receding-horizon
SCP scheme in which the nonlinear dynamics are linearized once at each
decision time. After the first optimized input is applied, the horizon
is shifted, the dynamics are relinearized about the new operating point,
and a new convex QP is solved. This provides a convenient numerical
formulation for handling input constraints while preserving the
Wasserstein tracking objective.

At decision time $k$, let $(x_k,u_k^{\mathrm{nom}})$ denote the nominal
state-input pair used for linearization. Given the definition of $\Psi$
in \eqref{eq:Psi}, define
\begin{equation}
\begin{aligned}
A_k
&:=
\frac{\partial\Psi}{\partial x}
(x_k,u_k^{\mathrm{nom}})
=
\frac{\partial f}{\partial x}(x_k)
+
\sum_{j=1}^{m}
\frac{\partial G_j}{\partial x}(x_k)
u_{k,j}^{\mathrm{nom}},\\
B_k
&:=
\frac{\partial\Psi}{\partial u}
(x_k,u_k^{\mathrm{nom}})
=
G(x_k).
\end{aligned}
\label{eq:scp_linearization_matrices}
\end{equation}

A first-order Taylor expansion of $\Psi$ about
$(x_k,u_k^{\mathrm{nom}})$ gives
\begin{equation}
x_{t+1}
\approx
\Psi(x_k,u_k^{\mathrm{nom}})
+
A_k(x_t-x_k)
+
B_k(u_t-u_k^{\mathrm{nom}}).
\label{eq:affine_linearization}
\end{equation}
Defining
\begin{equation}
\delta x_t:=x_t-x_k,
\qquad
\delta u_t:=u_t-u_k^{\mathrm{nom}},
\end{equation}
and
\begin{equation}
c_k:=\Psi(x_k,u_k^{\mathrm{nom}})-x_k,
\label{eq:linearization_offset}
\end{equation}
the local affine model becomes
\begin{equation}
\delta x_{t+1}
=
c_k+A_k\delta x_t+B_k\delta u_t,
\qquad
t=k,\ldots,k+H-1,
\label{eq:affine_linearized_dynamics}
\end{equation}
with $\delta x_k=0$. In general, $c_k\neq0$ because the
linearization point need not be an equilibrium of the nominal dynamics.

Note that the linearization is performed once at the current decision
point $(x_k,u_k^{\mathrm{nom}})$. The resulting $A_k$, $B_k$, and $c_k$
are then fixed throughout the $H$-step prediction horizon. This
fixed-model construction provides a single consistent linearized
prediction model for the local QP and enables explicit analysis of the
resulting closed-loop approximation error.

To separate the input-independent part of the affine prediction from
the input-dependent part, let $\delta x_t^{0}$ denote the prediction of
\eqref{eq:affine_linearized_dynamics} with $\delta u_t=0$. Thus,
\begin{equation}
\delta x_{t+1}^{0}
=
c_k+A_k\delta x_t^{0},
\qquad
\delta x_k^{0}=0.
\label{eq:zero_input_affine_prediction}
\end{equation}
Define the input-induced state deviation by
\begin{equation}
\delta\hat{x}_t
:=
\delta x_t-\delta x_t^{0}.
\end{equation}
Subtracting \eqref{eq:zero_input_affine_prediction} from
\eqref{eq:affine_linearized_dynamics} gives
\begin{equation}
\delta\hat{x}_{t+1}
=
A_k\delta\hat{x}_t+B_k\delta u_t,
\qquad
\delta\hat{x}_k=0.
\label{eq:increment_linear_dynamics}
\end{equation}
Since $y_t=Cx_t$, the corresponding input-induced output deviation is
\[
\delta\hat{y}_t=C\delta\hat{x}_t.
\]

Stacking these output and input deviations over the horizon, define
\begin{equation}
\begin{aligned}
\Delta Y
:=
\operatorname{col}
\left(
\delta\hat{y}_{k+1},\ldots,
\delta\hat{y}_{k+H}
\right),
\\
\Delta U
:=
\operatorname{col}
\left(
\delta u_k,\ldots,
\delta u_{k+H-1}
\right).
\end{aligned}
\end{equation}
Because $A_k$ and $B_k$ are fixed over the horizon, the standard
finite-horizon input-output relation gives
\begin{equation}
\Delta Y=\Theta\Delta U,
\label{eq:deltaY_theta_deltaU}
\end{equation}
where $\Theta$ is given by \eqref{eq:Theta}.

Further, defining
\begin{equation}
Y
=
\operatorname{col}
\left(
y_{k+1},\ldots,y_{k+H}
\right),
\;
Y^{0}
=
\operatorname{col}
\left(
y_{k+1}^{0},\ldots,y_{k+H}^{0}
\right),
\end{equation}
with
$
y_t^{0}
=
C\left(x_k+\delta x_t^{0}\right),
$
we obtain the affine prediction
\begin{equation}
Y
=
Y^{0}+\Theta\Delta U.
\label{eq:affine_output_prediction}
\end{equation}

Using $\Xi$ for the stacked Wasserstein barycenter in \eqref{eq:YQ},
the local tracking objective is constructed by
\begin{equation}
J(\Delta U)
=
\left\|
\boldsymbol{\Omega}
\left(
Y^{0}+\Theta\Delta U-\Xi
\right)
\right\|^2
+
\Delta U^\top\mathcal R\Delta U,
\label{eq:scp_tracking_objective}
\end{equation}
where
$
\mathcal R=I_H\otimes R$,
$
R\succ0.
$
The above cost $J$ corresponds to the least square cost in \eqref{eq:wasserstein_reduction} combined with the input penalty term.

Expanding \eqref{eq:scp_tracking_objective} with respect to $\Delta U$
gives
\begin{equation}
J(\Delta U)
=
\frac12\Delta U^\top\mathcal H\Delta U
+
g^\top\Delta U
+
\mathrm{const.},
\end{equation}
where the constant term is independent of $\Delta U$, and
\begin{equation}
\begin{aligned}
\mathcal H
&=
2\left[
(\boldsymbol{\Omega}\Theta)^\top
(\boldsymbol{\Omega}\Theta)
+
\mathcal R
\right],\\
g
&=
2(\boldsymbol{\Omega}\Theta)^\top
\boldsymbol{\Omega}
\left(
Y^{0}-\Xi
\right).
\end{aligned}
\label{eq:H and g}
\end{equation}
Since $\mathcal R\succ0$, it follows that
$\mathcal H\succ0$, establishing strict convexity of the local QP.

\begin{theorem}[Optimal Solution for Nonlinear $D^2$OC via Strictly Convex Local QP]
\label{thm:lin_qp}
At each SCP iteration, the local control update is obtained from the
quadratic program
\begin{equation}
\begin{aligned}
\min_{\Delta U}\quad
&
\frac12\Delta U^\top\mathcal H\Delta U
+
g^\top\Delta U\\
\mathrm{s.t.}\quad
&
u_{\min}^{(H)}
\leq
U^{\mathrm{nom}}+\Delta U
\leq
u_{\max}^{(H)},
\end{aligned}
\label{eq:local_qp}
\end{equation}
where $\mathcal H$ and $g$ are defined in \eqref{eq:H and g}, and
\begin{align*}
U^{\mathrm{nom}}
&:=
\mathbf{1}_H\otimes u_k^{\mathrm{nom}},\\
u_{\min}^{(H)}
&:=
\mathbf{1}_H\otimes u_{\min},
\quad
u_{\max}^{(H)}
:=
\mathbf{1}_H\otimes u_{\max}.
\end{align*}

Thus, the box constraints are imposed at every prediction stage
$t=k,\ldots,k+H-1$.

The local QP is strictly convex and therefore admits a unique global
minimizer whenever its constraints are feasible. If the unconstrained
minimizer is feasible, it is given by
\begin{equation}
\Delta U^{\mathrm{un}}
=
-\mathcal H^{-1}g.
\label{eq:local_qp_unconstrained_solution}
\end{equation}
Otherwise, the unique constrained minimizer is characterized by the KKT
conditions.
\end{theorem}

\begin{proof}
Let $M:=\Omega\Theta$. From \eqref{eq:H and g},
\[
\mathcal H
=
2M^\top M+2\mathcal R,
\qquad
\mathcal R=I_H\otimes R.
\]
Since $M^\top M\succeq0$ and $R\succ0$, we have
$\mathcal R\succ0$, and hence
$
\mathcal H\succ0.
$
Therefore, the QP objective is strictly convex while the box constraints
are affine. Consequently, whenever the feasible set is nonempty, the QP
admits a unique global minimizer.

If the unconstrained minimizer satisfies the box constraints, the first-order
condition
$
\mathcal H\Delta U^{\mathrm{un}}+g=0
$
directly gives \eqref{eq:local_qp_unconstrained_solution}.

For the constrained case, define the Lagrangian
\begin{align}
\mathcal L(\Delta U,\lambda^+,\lambda^-)
={}&
\frac12\Delta U^\top\mathcal H\Delta U
+g^\top\Delta U
\nonumber\\
&+
(\lambda^+)^\top
\left(
U^{\mathrm{nom}}+\Delta U-u_{\max}^{(H)}
\right)
\nonumber\\
&+
(\lambda^-)^\top
\left(
u_{\min}^{(H)}-U^{\mathrm{nom}}-\Delta U
\right).
\end{align}
The KKT conditions are
\begin{align}
\mathcal H\Delta U^\star+g+\lambda^+-\lambda^-&=0,
\label{eq:qp_kkt_stationarity}\\
\lambda^+,\lambda^-&\geq0,
\label{eq:qp_kkt_dual}\\
u_{\min}^{(H)}
\leq U^{\mathrm{nom}}+\Delta U^\star
&\leq u_{\max}^{(H)},
\label{eq:qp_kkt_primal}\\
\lambda^+\odot
\left(
U^{\mathrm{nom}}+\Delta U^\star-u_{\max}^{(H)}
\right)
&=0,
\label{eq:qp_kkt_upper_complementarity}\\
\lambda^-\odot
\left(
u_{\min}^{(H)}-U^{\mathrm{nom}}-\Delta U^\star
\right)
&=0.
\label{eq:qp_kkt_lower_complementarity}
\end{align}
Since the objective is convex and the constraints are affine, these KKT
conditions are necessary and sufficient for global optimality. Strict
convexity of the objective further guarantees uniqueness of
$\Delta U^\star$.
\end{proof}

\begin{remark}[Significance of the SCP-QP Formulation]\label{rem:scp_contribution}
While the general iterative solution mechanism aligns with SCP, Theorem~\ref{thm:lin_qp} provides a non-trivial reduction tailored specifically to the $D^2OC$ framework. Rather than applying generic Sequential Quadratic Programming (SQP), this formulation explicitly maps the Wasserstein barycenter tracking metrics into the Hessian $\mathcal H$ and gradient $g$ matrices. This structured reduction guarantees strict convexity ($\mathcal H \succ 0$) and uniqueness of $\Delta U^\star$ at each step, serving as the precise computational backbone for the rigorous error bound and certificate analysis developed in the sequel.
\end{remark}

\section{Linearization Error and Finite-Horizon Error Certification}
\label{sec:error_analysis}

The local QP in Section~\ref{sec:opt_nonlin} is constructed from a first-order
linearization of the nonlinear dynamics. Although this reduction yields a convex
optimization problem, the resulting prediction can deviate from the nonlinear
trajectory due to both input clipping and the higher-order Taylor remainder of the
dynamics. This section quantifies these effects and establishes a finite-horizon
bound on the resulting linearization error.
The following assumptions are provided to analyze this linearization error. 
\begin{assumption}[Regularity and curvature bounds]\label{assum:regularity}
$f,G\in C^2$ on a convex set $\mathcal X\times\mathcal U\subset\mathbb R^n\times\mathbb R^m$
containing the relevant state-input trajectories, with
$
\|f_{xx}(x)\|\le M_f,\;\|\partial_{xx}G_j(x)\|\le M_{Gx},\; \|\partial_x G_j(x)\|\le M_{Gu},
\;\forall (x,u)\in\mathcal X\times\mathcal U.
$
Set $\bar M:=M_f+mM_{Gx}$ and $\bar M':=mM_{Gu}$, where $m$ denotes the input dimension.
\end{assumption}
This regularity condition ensures that the second-order Taylor remainder of $\Psi$ is uniformly bounded over $\mathcal X$. Furthermore, the convexity of $\mathcal X$ guarantees that line segments connecting candidate trajectories remain inside $\mathcal X$, enabling the application of the mean-value theorem.

\begin{assumption}[Uniform sensitivity on the operative tube]
\label{assum:sensitivity}
Let $\Psi(x,u)$ be defined as in \eqref{eq:Psi}. There exists a finite
constant $L>0$ such that
\[
\|\partial_x\Psi(x,u)\|\le L,
\qquad
\forall (x,u)\in\mathcal X\times\mathcal U.
\]
Moreover, define
\[
\kappa:=\sup_{x\in\mathcal X}\|G(x)\|<\infty.
\]
Then,
$
\|\partial_u\Psi(x,u)\|
=\|G(x)\|
\le\kappa,
\quad
\forall (x,u)\in\mathcal X\times\mathcal U.
$
\end{assumption}

The symbol $\kappa$ limits the single-step dynamic perturbation from input clipping, while the Lipschitz bound $L$ permits the propagation of this perturbation across the prediction horizon via a discrete Gr\"onwall inequality \cite{pachpatte1998inequalities}.

\begin{assumption}[Small nominal excursion]
\label{assum:small-excursion}
There exist constants $\varepsilon_x>0$ and $\varepsilon_u>0$ such that
\[
\|\delta x_t^{\mathrm{un}}\| \le \varepsilon_x,
\quad
\|\delta u_t^{\mathrm{un}}\| \le \varepsilon_u,
\quad
t = k, \dots, k+H-1.
\]
\end{assumption}
The above assumption ensures that the nominal trajectory remains within the
local neighborhood in which the subsequent sensitivity bounds are valid.

\begin{lemma}[Uniform control deviation bound]
\label{lemma:control_deviation}
Let $\Delta U^\star$ and $\Delta U^{\mathrm{un}}$ denote the
constrained and unconstrained minimizers of the local QP in
Theorem~\ref{thm:lin_qp}. Define
\begin{equation}
\delta_u^{\max}
:=
\left\|
\Delta U^\star-\Delta U^{\mathrm{un}}
\right\|.
\label{eq:delta_u_max}
\end{equation}
Then, for every prediction stage
$t=k,\ldots,k+H-1$,
\begin{equation}
\left\|
u_t^\star-u_t^{\mathrm{un}}
\right\|
\leq
\delta_u^{\max}.
\label{eq:stage_control_deviation}
\end{equation}
\end{lemma}

\begin{proof}
Recall that
$\delta u_t^\star
:=
u_t^\star-u_k^{\mathrm{nom}}$,
$
\delta u_t^{\mathrm{un}}
:=
u_t^{\mathrm{un}}-u_k^{\mathrm{nom}}.
$
Since $\Delta U^\star$ and $\Delta U^{\mathrm{un}}$ are stacked as
\begin{equation}
    \begin{aligned}
\Delta U^\star
&=
\operatorname{col}
\left(
\delta u_k^\star,\ldots,\delta u_{k+H-1}^\star
\right),
\\
\Delta U^{\mathrm{un}}
&=
\operatorname{col}
\left(
\delta u_k^{\mathrm{un}},\ldots,\delta u_{k+H-1}^{\mathrm{un}}
\right),
    \end{aligned}
\end{equation}
we have
$
u_t^\star-u_t^{\mathrm{un}}
=
\delta u_t^\star-\delta u_t^{\mathrm{un}}.
$
Because $\delta u_t^\star-\delta u_t^{\mathrm{un}}$ is a block of
$\Delta U^\star-\Delta U^{\mathrm{un}}$, the Euclidean norm satisfies
\[
\left\|
u_t^\star-u_t^{\mathrm{un}}
\right\|
=
\left\|
\delta u_t^\star-\delta u_t^{\mathrm{un}}
\right\|
\leq
\left\|
\Delta U^\star-\Delta U^{\mathrm{un}}
\right\|
=
\delta_u^{\max}.
\]
\end{proof}

The preceding development establishes the local SCP subproblem and
quantifies the control deviation introduced by the input constraints. We next use this control deviation as
the starting point for analyzing the accuracy of the local linearization
along the resulting nonlinear trajectory.

\begin{theorem}[Linearization error under QP clipping]
\label{thm:err_with_clip}
Under Assumptions~\ref{assum:regularity}--\ref{assum:small-excursion}, let
$A,B$ denote the local linearization matrices of Theorem~\ref{thm:lin_qp}
at $(x_k,u_k^{\mathrm{nom}})$, with $\|A\|\le\alpha$.
Let $\Delta U^{\mathrm{un}}$ and $\Delta U^\star$ be the unconstrained and
box-constrained QP solutions, respectively, and let
$x_t^{\mathrm{un}}$ and $x_t^\star$ denote the corresponding nonlinear
trajectories initialized at the common state $x_k$.

Define
$
\delta x_t^\star:=x_t^\star-x_k$ and
$
\delta u_t^\star:=u_t^\star-u_k^{\mathrm{nom}},
$
and let $x_t^{\mathrm{lin}}$ be the state predicted by the local linear
model under the \emph{applied constrained input},
$$\delta x_{t+1}^{lin}=c_k+A\delta x_t^{lin}+B\delta u_t^\star,\qquad\delta x_k^{lin}=0,$$
where $c_k$ is described in \eqref{eq:linearization_offset}.
Set $e_t:=x_t^\star-x_t^{\mathrm{lin}}$, $e_k=0$, and let $\delta_u^{\max}$
be as in \eqref{eq:delta_u_max}. Define
\[
\bar\varepsilon_x
:=
\varepsilon_x+
\kappa\delta_u^{\max}
\sum_{j=0}^{N-k-1}L^j,
\qquad
\bar\varepsilon_u
:=
\varepsilon_u+\delta_u^{\max}.
\]
Then, for every $k<N\le k+H$,
\begin{equation}
\label{eq:err_clip_final}
\|e_N\|
\le
\Big(
\sum_{j=0}^{N-k-1}\alpha^j
\Big)
\left[
\frac12\bar M\bar\varepsilon_x^2
+
\bar M'\bar\varepsilon_x\bar\varepsilon_u
\right].
\end{equation}
Equivalently, if $\alpha\neq1$ and $L\neq1$,
\[
\bar\varepsilon_x
=
\varepsilon_x
+
\kappa\delta_u^{\max}
\frac{L^{N-k}-1}{L-1},
\]
and
\[
\|e_N\|
\le
\frac{\alpha^{N-k}-1}{\alpha-1}
\left[
\frac12\bar M\bar\varepsilon_x^2
+
\bar M'\bar\varepsilon_x\bar\varepsilon_u
\right].
\]
If $\alpha=1$ or $L=1$, the corresponding geometric sum is replaced
by $N-k$.
\end{theorem}

\begin{proof}
Define the clipping-induced trajectory deviation
\[
\hat e_t:=x_t^\star-x_t^{\mathrm{un}},
\qquad
\hat e_k=0.
\]
Since both trajectories remain in $\mathcal X\times\mathcal U$ by
Assumption~\ref{assum:regularity}, the line segment connecting
$(x_t^{\mathrm{un}},u_t^{\mathrm{un}})$ and
$(x_t^\star,u_t^\star)$ lies in $\mathcal X\times\mathcal U$.
Parameterize this segment by
\[
x_s=x_t^{\mathrm{un}}+s\hat e_t,
\qquad
u_s=u_t^{\mathrm{un}}
+s(u_t^\star-u_t^{\mathrm{un}}),
\qquad s\in[0,1].
\]
By the fundamental theorem of calculus,
\[
\begin{aligned}
\hat e_{t+1}
&=
\Psi(x_t^\star,u_t^\star)
-\Psi(x_t^{\mathrm{un}},u_t^{\mathrm{un}})\\
&=
\int_0^1
\left[
\partial_x\Psi(x_s,u_s)\hat e_t
+
\partial_u\Psi(x_s,u_s)
(u_t^\star-u_t^{\mathrm{un}})
\right]ds.
\end{aligned}
\]
Therefore, using
$\|\partial_x\Psi\|\le L$,
$\|\partial_u\Psi\|\le\kappa$, and
$\|u_t^\star-u_t^{\mathrm{un}}\|\le\delta_u^{\max}$,
\[
\|\hat e_{t+1}\|
\le
L\|\hat e_t\|
+
\kappa\delta_u^{\max}.
\]
Since $\hat e_k=0$, iteration yields
$
\|\hat e_t\|
\le
\kappa\delta_u^{\max}
\sum_{j=0}^{t-k-1}L^j.
$
Hence, for all $t\le N$,
\[
\|\hat e_t\|
\le
\kappa\delta_u^{\max}
\sum_{j=0}^{N-k-1}L^j.
\]

By Assumption~\ref{assum:small-excursion},
$
\|\delta x_t^{\mathrm{un}}\|\le\varepsilon_x.
$
Thus,
\[
\|\delta x_t^\star\|
\le
\|\delta x_t^{\mathrm{un}}\|
+
\|\hat e_t\|
\le
\bar\varepsilon_x.
\]
Likewise,
\[
\|\delta u_t^\star\|
\le
\|\delta u_t^{\mathrm{un}}\|
+
\|u_t^\star-u_t^{\mathrm{un}}\|
\le
\bar\varepsilon_u.
\]

We now expand the nonlinear dynamics about the fixed linearization point
$(x_k,u_k^{\mathrm{nom}})$. Since
$
\Psi(x,u)=f(x)+G(x)u
$
is affine in $u$, its second derivative with respect to $u$ vanishes.
Consequently, the second-order Taylor remainder contains only
state-state and state-input terms. By the Hessian bounds in
Assumption~\ref{assum:regularity},
\[
\|r_t\|
\le
\frac12\bar M\|\delta x_t^\star\|^2
+
\bar M'\|\delta x_t^\star\|\|\delta u_t^\star\|
\le
\frac12\bar M\bar\varepsilon_x^2
+
\bar M'\bar\varepsilon_x\bar\varepsilon_u.
\]
Define
$
\eta
:=
\frac12\bar M\bar\varepsilon_x^2
+
\bar M'\bar\varepsilon_x\bar\varepsilon_u.
$

Because the linear prediction uses the same applied input
$\delta u_t^\star$ as the nonlinear trajectory, subtraction gives the
exact error recursion
$
e_{t+1}=Ae_t+r_t$,
$
e_k=0.
$

Therefore,
$
e_N
=
\sum_{t=k}^{N-1}A^{N-1-t}r_t,
$
and hence
\[
\begin{aligned}
\|e_N\|
&\le
\sum_{t=k}^{N-1}
\|A\|^{N-1-t}\|r_t\|\\
&\le
\eta
\sum_{t=k}^{N-1}\alpha^{N-1-t}
=
\eta
\sum_{j=0}^{N-k-1}\alpha^j.
\end{aligned}
\]
Substituting the definition of $\eta$ proves
\eqref{eq:err_clip_final}.
\end{proof}

\begin{corollary}[Two-step specialization of the finite-horizon error bound]
\label{cor:first_nontrivial}
Under the conditions of Theorem~\ref{thm:err_with_clip}, consider
$N=k+2$. Then the linearization error after two prediction steps satisfies
\begin{equation}
\begin{aligned}
\|e_{k+2}\|
&\leq
\frac12\bar M
\left(
\varepsilon_x+\kappa\delta_u^{\max,(0)}
\right)^2 \\
&\quad+
\bar M'
\left(
\varepsilon_x+\kappa\delta_u^{\max,(0)}
\right)
\left(
\varepsilon_u+\delta_u^{\max}
\right).
\end{aligned}
\label{eq:first_nontrivial_bound}
\end{equation}

In particular, since $\Psi$ is affine in $u$ and
$\delta x_k^\star=0$, the first-step Taylor remainder satisfies
$r_k=0$, and hence $e_{k+1}=0$. Thus, the first nonzero
linearization error can occur at $N=k+2$.
\end{corollary}

\begin{proof}
At $t=k$, $\delta x_k^\star=0$. Since $\Psi$ is affine in $u$, the
second-order Taylor remainder satisfies $r_k=0$, and therefore
\[
e_{k+1}=Ae_k+r_k=0.
\]
Hence, the error at $k+2$ is generated only by the Taylor remainder
at the second step. Applying the local remainder bound at
$t=k+1$ and using the one-step state and input deviations gives \eqref{eq:first_nontrivial_bound}, where
$\delta_u^{\max,(0)}$ denotes the clipping magnitude of the first
input block, with $\delta_u^{\max,(0)}\leq\delta_u^{\max}$.
\end{proof}

\begin{remark}[Corollary~\ref{cor:first_nontrivial} sharpens Theorem~\ref{thm:err_with_clip} at $N=k+2$]
\label{rem:first_nontrivial_scope}
Corollary~\ref{cor:first_nontrivial} targets the same quantity, $e_{k+2}$,
that Theorem~\ref{thm:err_with_clip} bounds at $N=k+2$, but exploits the
specific initialization $e_k=0$ and the fact that $r_k=0$. Consequently,
the Gr\"onwall recursion is needed only once, from $k$ to $k+1$, rather
than over the entire horizon. This removes the $L$-dependent factor
$1+L$ from the state envelope $\bar\varepsilon_x$ appearing under direct
substitution of $N=k+2$ into \eqref{eq:err_clip_final}: only the first
input block contributes to $\hat e_{k+1}$, which is captured by
$\delta_u^{\max,(0)}\leq\delta_u^{\max}$. The input envelope
$\bar\varepsilon_u$, by contrast, still bounds $\delta u_{k+1}^\star$ --
the \emph{second} input block -- which $\delta_u^{\max,(0)}$ does not
control; it therefore retains the horizon-wide quantity $\delta_u^{\max}$
from Lemma~\ref{lemma:control_deviation}. Thus,
\eqref{eq:first_nontrivial_bound} provides a bound that is sharper than
direct substitution of $N=k+2$ into Theorem~\ref{thm:err_with_clip} in its
state-dependent term, while matching it in its input-dependent term.
\end{remark}

\begin{remark}[Implications for Closed-Loop Stability]
\label{rem:stability_scope}
This bound characterizes local prediction accuracy rather than
closed-loop stability or asymptotic convergence. Specifically,
\eqref{eq:first_nontrivial_bound} quantifies the two-step deviation of
the actual nonlinear state from the linearized model prediction under
the applied input, without implying state boundedness or target
tracking under $\mathrm{D}^2\mathrm{OC}$. Thus,
Corollary~\ref{cor:first_nontrivial} serves as the foundational
disturbance bound required for such analysis rather than a standalone
stability proof.
\end{remark}

\section{Simulation}\label{sec:sim}

We evaluate the proposed controller on two control-affine nonlinear multi-agent
systems: a unicycle team and a quadrotor team. The proposed SCP procedure
relinearizes the dynamics at each decision step and solves the resulting
strictly convex local QP, as characterized in
Theorem~\ref{thm:lin_qp}. We compare its coverage performance and
computational cost with a nonlinear NMPC baseline under the same objective,
horizon, and input constraints. The simulation also evaluates the
linearization error bound established in Theorem~\ref{thm:err_with_clip}.

\subsection{Unicycle Model}\label{sec:sim_uni}

Five agents with state $x=[p_x,p_y,\theta]^\top$ and input
$u=[v,\omega]^\top$ cover a $100\times100$ m domain. The sampling time
and prediction horizon are set to $T=0.1$ s and $H=20$, respectively.
Only the first input of each horizon is applied before relinearization. The
reference distribution $\nu$ contains $500$ samples drawn from a
five-component Gaussian mixture, as shown in Fig.~\ref{fig:uni_traj}.
The initial positions of the five unicycle agents are indicated by cross
marks. To evaluate the proposed method, ten simulation runs are performed,
with the reference sample points independently redrawn for each of the ten
seeds. The sample weights are initially assigned uniformly as
$\frac{1}{500}$ and then decreased according to the rule in
\cite{seo2025smcs} as the agents cover the domain. The simulation is
terminated when the sum of the remaining sample weights falls below one
percent of its initial value.

\begin{figure}[!h]
\centering
\includegraphics[width=\columnwidth]{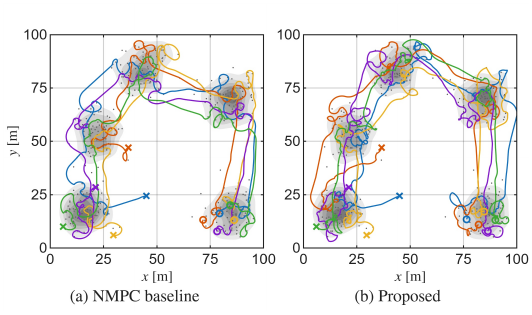}
\caption{Five-unicycle team trajectories: (a) NMPC
baseline and (b) proposed method. Reference samples are shaded by weight;
crosses and circles denote initial and final positions.}
\label{fig:uni_traj}
\end{figure}

Fig. \ref{fig:uni_traj} shows that both the NMPC and SCP controllers
effectively cover the domain according to the prescribed reference density.
The Wasserstein distance in Fig. \ref{fig:uni_perf}(a), which quantifies
the discrepancy between the time-averaged agent distribution and the
reference distribution, further illustrates the density-driven coverage
performance. The mean $\mathcal{W}_2(\rho_k,\nu)$ over the final ten
evaluations is $5.28$ for the proposed method and $5.23$ for NMPC,
with no significant difference in variance observed. These results indicate
that the linearized update based on Theorem~\ref{thm:lin_qp} preserves
coverage performance comparable to the nonlinear baseline in this
simulation.

\begin{figure}[!h]
\centering
\includegraphics[width=\columnwidth]{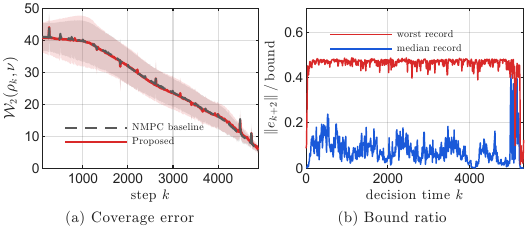}
\caption{Unicycle team over ten seeds: (a) Wasserstein distance between the
time-averaged agent distribution $\rho_k$ and reference density $\nu$ and
(b) Ratio of $||e_{k+2}||$ to the proposed bound.}
\label{fig:uni_perf}
\end{figure}

Fig. \ref{fig:uni_perf}(b) evaluates the error bound in
Theorem\ref{thm:err_with_clip} using \eqref{eq:first_nontrivial_bound}.
Since only the first input of each horizon is applied before relinearization,
we evaluate the first nonzero Taylor-remainder case at $N=k+2$, corresponding
to Corollary~\ref{cor:first_nontrivial}. Across the ten seeds, the worst-case
and median values are plotted over the simulation steps. The ratio
$|e_{k+2}|/\mathrm{bound}$ remains below one for all evaluated steps,
consistent with the bound in Corollary~\ref{cor:first_nontrivial}.

At each solve, $\varepsilon_x$ and $\varepsilon_u$ are instantiated as the
largest state and input excursions, respectively, of the unconstrained plan
over $t=k,\dots,k+H-1$. Thus, Assumption~\ref{assum:small-excursion} is
verified on the generated trajectory rather than imposed a priori. The
constants $\kappa$, $L$, $\bar M$, and $\bar M'$ are instantiated
analogously. Rather than fixing them a priori as uniform bounds over a
pre-specified $\mathcal X\times\mathcal U$, we evaluate $\kappa$ and $L$
using the maximum Jacobian norms of $\Psi$ sampled along the realized pairs
$(x_t^\star,u_t^\star)$ and $(x_t^{\mathrm{un}},u_t^{\mathrm{un}})$ at each
decision time $k$. The quantities $\bar M$ and $\bar M'$ are obtained from
finite-difference estimates of the sensitivity of the same Jacobians with
respect to $x$ over the corresponding samples. Accordingly,
Assumptions~\ref{assum:regularity} and~\ref{assum:sensitivity} are evaluated
on the realized operative tube at each $k$, rather than as fixed uniform
bounds over a pre-specified region.

To evaluate the computational advantage of the proposed method, the
computation time is measured for each run. The median computation time
across the ten runs is $1.61$ ms for the proposed SCP method and $9.20$ ms
for NMPC, corresponding to a factor of $5.7$ reduction. This result
demonstrates the computational benefit of replacing the nonlinear
optimization with the strictly convex QP characterized in
Theorem~\ref{thm:lin_qp}.

\subsection{Quadrotor Model}\label{sec:sim_quad}

We next evaluate the method on four quadrotors with a full twelve-state
rigid-body model in a $100\times100\times100$ m domain. The input consists
of collective thrust and body torques. The sampling time and prediction
horizon are set to $T=0.05$ s and $H=12$, respectively. The reference
distribution contains $600$ sample points drawn from a four-component
Gaussian mixture, as depicted in Fig.~\ref{fig:quad_traj}, with the samples
independently redrawn for each of the ten seeds. The initial positions of
the four quadrotors are indicated by cross marks.

Although the density-driven coverage trajectories of the NMPC and SCP
methods in Fig. \ref{fig:quad_traj} differ, the quantitative measure in
Fig. \ref{fig:quad_perf}(a) shows that both controllers achieve essentially
the same coverage level. The averages over the final ten evaluations are
$12.74$ for both methods, with a difference of only $0.003$. These
results indicate that the proposed SCP-based linearization can drive the
nonlinear multi-agent system toward the prescribed reference distribution
while maintaining comparable coverage performance to the NMPC baseline.

\begin{figure}[tb]
\centering
\includegraphics[width=\columnwidth]{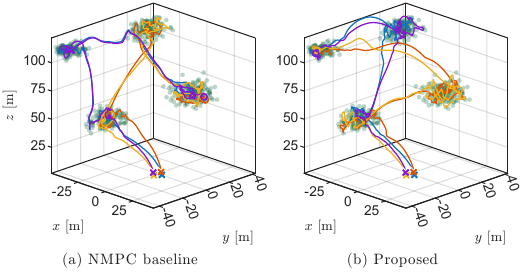}
\caption{Four-quadrotor team trajectories: (a) NMPC
baseline and (b) proposed method. Crosses and circles denote initial and
final positions.}
\label{fig:quad_traj}
\end{figure}

\begin{figure}[tb]
\centering
\includegraphics[width=\columnwidth]{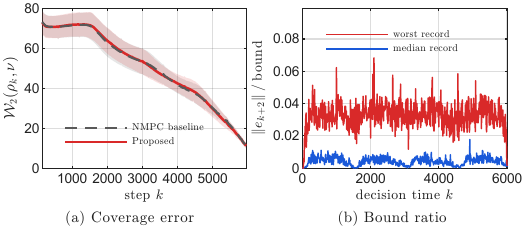}
\caption{Quadrotor team over ten seeds: (a) Wasserstein distance between
$\rho_k$ and $\nu$, and (b) Ratio of $||e_{k+2}||$ to the proposed bound.}
\label{fig:quad_perf}
\end{figure}

Fig. \ref{fig:quad_perf}(b) further presents the worst-case and median
values of the bound ratio across the ten simulation runs, supporting the
error-bound analysis in Theorem\ref{thm:err_with_clip}. The parameters
$\kappa$, $L$, $\bar M$, and $\bar M'$ are again instantiated locally at
each decision time using the same procedure as for the unicycle model.

Finally, the computational advantage becomes more pronounced for the
quadrotor case due to its higher state dimension. The median computation
time is $3.64$ ms for the proposed SCP method versus $162.2$ ms for NMPC,
corresponding to a factor of $45$, compared with a factor of $5.7$ for the
unicycle case.

\section{Conclusion}
\label{sec:conclusion}

This paper extended D$^2$OC to discrete-time control-affine nonlinear multi-agent systems using sequential convex programming. The proposed approach preserves the Wasserstein-based coverage objective and barycentric structure while reducing each local update to a strictly convex QP with input constraints. We further quantified the deviation between the constrained and unconstrained local QP solutions and incorporated this control deviation into a finite-horizon bound on the error between the nonlinear trajectory and its linearized prediction. Simulations with unicycle and quadrotor teams showed that the proposed method achieves coverage performance comparable to NMPC while substantially reducing computation time. These results demonstrate the potential of the proposed formulation for computationally efficient nonlinear coverage control. Future work will focus on closed-loop stability, recursive feasibility, and extensions to more general nonlinear constraints.

\bibliographystyle{ieeetr}
\bibliography{references}

\end{document}